\documentclass[11pt,reqno]{amsart}
\usepackage[utf8]{inputenc}
\usepackage{a4wide}
\usepackage[T1]{fontenc}
\usepackage{amssymb}
\usepackage[english]{babel}
\usepackage{a4wide}
\usepackage{xcolor}
\usepackage{enumitem}
\usepackage{bbm}
\usepackage[pagebackref, colorlinks = true, linkcolor = teal, urlcolor  = teal, citecolor = violet]{hyperref}
\usepackage{upgreek}
\usepackage{graphicx}

\newtheorem{thm}{Theorem}[section]
\newtheorem{prop}[thm]{Proposition}

\newtheorem{cor}[thm]{Corollary}
\theoremstyle{plain}
\newtheorem{lem}[thm]{Lemma}
\theoremstyle{definition}

\theoremstyle{remark}

\newtheorem{rmq}{Remark}
\numberwithin{equation}{section}

\newcommand{\tr}{\operatorname{tr}}

\newcommand{\R}{\mathbb{R}}

\newcommand{\pp}{\mathbb{P}}
\newcommand{\pprho}[1]{\pp^{#1}}

\newcommand{\ee}{\mathbb{E}}
\newcommand{\eerho}[1]{\ee^{#1}}

\newcommand{\cD}{\mathcal{D}}
\newcommand{\cF}{\mathcal{F}}

\newcommand{\ii}{\mathrm{i}}
\newcommand{\dd}{\mathrm{d}}
\newcommand{\e}{\mathrm{e}}
\newcommand{\Lin}{\mathcal{L}}

\newcommand{\Id}{\mathrm{Id}}
\renewcommand{\and}{\mbox{ and }}

\newcommand{\filtreP}{\big(\Omega,(\cF_t)_t,\cF_\infty,\pp\big)}

\usepackage{xargs}

\title{On stability of quantum trajectories and their Cesaro mean}
\author{Nina H. Amini , Ma\"el Bompais and Cl\'ement Pellegrini}

\thanks{Nina H. Amini and Ma\"{e}l Bompais are with Laboratoire des Signaux et Syst\`emes, CNRS - CentraleSup\'elec - Univ. Paris-Sud, Universit\'e Paris-Saclay, 3, rue Jo liot Curie, 91192, Gif-sur-Yvette, France.
 {\tt\small first name.family name@centralesupelec.fr}}
\thanks{Cl\'ement Pellegrini is with Institut de Math\'ematiques, IMT, Universit\'e de Toulouse (UMR 5219), 31062 Toulouse, Cedex 9, France 
        {\tt\small clement.pellegrini@math.univ-toulouse.fr}}%

\begin{document}
	
\begin{abstract}{We address the question of stability of quantum trajectories, also referred as quantum filters}. We determine the limit of the quantum fidelity between the true quantum trajectory and the {estimated one}. Under a purification assumption we show that this limit equals to one meaning that quantum filters are stable. In the general case, under an identifiability and a spectral assumption we show that the limit of the Cesaro mean of the estimated trajectory is the same as the true one.
\end{abstract}

\maketitle

\setcounter{tocdepth}{1}
\tableofcontents
\section{Introduction}

Quantum trajectories describe the evolution of an open quantum system \cite{Davies} undergoing indirect measurements \cite{BRE02,wisemanmilburn,Haroche,BarchielliGregoratti09,gisin,diosi,belavkin89,bouten}. {The setup consists of a quantum system interacting with an environment.} After interaction, the environment is observed and following the measurement record, one can infer the evolution of the system. Typically, the evolution of $(\rho_t)$, i.e., a quantum trajectory of the quantum system is random taking into account the back-action of the measurement. Generic models describing $(\rho_t)$ are stochastic differential equations called stochastic master equations driven by jump-diffusion processes \cite{barchielliholevo,BarchielliGregoratti09,pellegrini3}. In absence of measurements, the evolution of the system is described by a Lindblad operator \cite{Davies,lindblad,gorini}. 

 Recently these models have attracted a lot of investigations studying the large time behavior of quantum trajectories. In \cite{bauer2011convergence,bbbqnd} models of quantum non demolition which are at the cornerstone of the recent experiment of Serge Haroche's team are studied in details. In particular in \cite{bauer2011convergence,bbbqnd,Bauer_2012}, conditioning and martingale techniques are used to obtain sharp results concerning the large time behavior. These techniques have been used also in the continuous time setup in \cite{MR3238528}. Generalization of these approaches are developed in \cite{MR3768183} for quantum parameter estimation and in progress in \cite{BBP2020} {concerning a law of large number, a central limit theorem and a large deviation principle}. In \cite{MR3947325,MR4205230}, results on invariant measure and convergence toward the stationary regime {are} obtained under a purification 
and an ergodicity assumptions. Recent contributions, namely entropy production considerations, are developed in \cite{MR3764564,MR4222567} using ergodic theory and thermodynamic formalism. {It is worth noticing} that the first relevant results in the long time behavior of quantum trajectories were the purification result \cite{HMpuri} and an ergodic theorem \cite{HMergo}. 


In control theory, quantum trajectories are referred as quantum filters. The theory has been initiated by E. B. Davies in the 1960s~\cite{davies1969quantum,davies1976quantum} and further developed by Belavkin in the 1980s~\cite{belavkin1983theory,belavkin1989nondemolition,belavkin1995quantum,belavkin1992quantum}. The quantum filters are applied in the design of a state-based feedback. In particular, stabilization of pure states by a state-based feedback is central in advancing quantum technologies. In~\cite{sayrin2011real,amini2012stabilization,amini2013feedback}, stabilizations of quantum systems undergoing discrete-time non-demolition measurements are considered. Regarding stabilization of continuous-time quantum filters, see e.g.,~\cite{mirrahimi2007stabilizing,van2005feedback,liang2019exponential,cardona2020exponential,MR3649451,MR3006708,9375489}.   Roughly speaking, quantum filters should take into account different physical imperfections such as unknown initial states and physical parameters, which are present in real experiments. In~\cite{liang2020robustness_N}, the robustness of feedback strategies proposed in~\cite{liang2019exponential} with respect to such imperfections has been addressed. 

In this article, we investigate stability of quantum filters. When the initial state $\rho_0$ of the quantum trajectory $(\rho_t)$ is unknown, we construct an estimated filter by guessing the initial state $\hat\rho_0$. Then one uploads an estimated trajectory $(\hat\rho_t)$ following the results of the measurements which are the only data accessible. This estimated trajectory $(\hat\rho_t)$ evolves as if it was the true trajectory $(\rho_t)$ but starting with a guessed initial state. The natural question is whether the trajectories become closer and closer by acquiring more and more results after measurements. Does the distance between the true trajectory and the estimated trajectory converges to zero in long time? Convergence results have been already investigated. In \cite{MR2509990}, it was shown that observability provides sufficient condition to ensure stability of quantum filters. In \cite{bbbqnd,MR3238528}, convergence for quantum non demolition models are obtained. The natural distance in this context is the fidelity and in \cite{MR3255126,MR2884021}, it has been shown  that in average this fidelity increases.\footnote{The fidelity $F(\rho,\sigma)$ satisfies that $0\leq F(\rho,\sigma)\leq 1$ and $F(\rho,\sigma)=1$ if and only if $\rho=\sigma$} In this article, we study the limit of the fidelity between $(\rho_t)$ and $(\hat\rho_t)$, we show that this fidelity can be expressed in terms of a particular martingale which is convergent. This martingale expressed in terms of the operator generating the outputs of measurements, is at the cornerstone of the works \cite{MR3947325,MR4205230}. In particular we obtain an expression for the limit of this fidelity when $t$ goes to infinity. This allows to show that under a purification condition this limit of fidelity is one meaning that the quantum filters are stable. We show that it is not the case in general. In the general case, where the limit fidelity is not one we investigate a mean ergodic theorem. For a true trajectory it is known that the Cesaro mean converges toward a random variable with values on the set of invariant states of the evolution without measurement \cite{HMergo}. We address the same question for the estimated trajectory. {Under an absolutely continuity condition, we show that the Cesaro mean of estimated trajectory converges also toward an invariant state.} When there are more than one invariant state the limit of the Cesaro mean for the estimated trajectory is in general different than the one of the Cesaro mean of the true trajectory. It is then natural to present sufficient conditions where the two limits coincide. Under an identifiability assumption as well as a spectral assumption we show that the two limits coincide. This uses the conditioning techniques of \cite{bbbqnd}.

The article is structured as follows. Section 2 is devoted to the presentation of the true trajectory and the estimated one. In particular, {we present an absolute continuity between initial states which is an essential ingredient of the paper.} Section 3 concerns the fidelity between the true and estimated trajectories. We study the limit of this quantity when $t$ goes to infinity. Section 4 is devoted to the Cesaro mean considerations. Some technical proofs are devoted into an appendix.

\section{Quantum Trajectories}\label{model}

\subsection{Construction of quantum trajectories}

In this section, we recall some basic construction of continuous time quantum trajectories \cite{BarchielliGregoratti09}. The underlying quantum system is described by $\mathcal H=\mathbb C^k$ and the set of density matrices is denoted by $$\mathcal D_k=\{\rho\in\mathcal M_k(\mathbb C), \rho=\rho^*,\rho\geq0,\tr(\rho)=1\}.$$ Within the construction of quantum trajectories, we present the evolution of the estimated filter. The true quantum trajectory will be denoted by $(\rho_t)$ whereas the estimated filter will be denoted by $(\hat \rho_t)$.

We consider a filtered probability space $\filtreP$. On this space we consider $(S_t)_{t\in[0,\infty)}$ the solution to the following stochastic differential equation (SDE for short):
\begin{equation}\label{eq_defS}
\dd S_{t}= \big(K+{\textstyle {\frac{n-p}2}}\,{\Id}\big)S_{t-}\,\dd t+\sum_{i=1}^p L_iS_{t-}\,\dd W_i(t)+\sum_{j=p+1}^n(C_j-\Id)S_{t-}\,\dd N_j(t),\qquad S_0=\Id
\end{equation}
 where $$K=-\ii H-\frac{1}{2}\left(\sum_{i=1}^p L_i^*L_i+\sum_{j=p+1}^n C_j^*C_j\right).$$
 The operator $H$ is an Hermitian operator which plays the role of an Hamiltonian and the operators $L_i,i=1,\ldots,p$, $C_j,j=p+1,\ldots,n$ are any operators. The parameter $n$ is a fixed integer which corresponds to the number of noises. On $\filtreP$, the processes $(W_i(t)),i=1,\ldots,p$, appearing in the above equation, are independent brownian motions, independent of independent Poisson processes $(N_j(t)),j=p+1,\ldots,n$.

Now, we shall introduce the process which corresponds to the measurement records. For any $\rho\in\cD_k$, let $(Z_t^\rho)_t$ be the positive real-valued process defined by
\[Z_t^\rho=\tr(S_t^*S_t\rho),\]
and let $(\rho_t)_{t}$ be the $\cD_k$-valued process defined by
\begin{equation}\label{eq:qtraj}\rho_t=\frac{S_t\rho S_t^*}{\tr(S_t\rho S_t^*)}\end{equation}
if $Z_t^\rho\neq 0$, taking an arbitrarily fixed value whenever $Z_t^\rho=0$. The process $(\rho_t)$ is called a quantum trajectory with initial condition $\rho_0=\rho$

The following results on the properties of $(Z_t^\rho)_t$ were proven in \cite{barchielliholevo}.
\begin{lem}\label{lem_Zmg} For any $\rho\in\cD_k$, the stochastic process $(Z_t^\rho)_t$ is the unique solution of the SDE
\[\dd Z_t^\rho=Z_{t-}^\rho\Big(\sum_{i=1}^p\tr\big((L_i+L_i^*)\rho_{t-}\big)\dd W_i(t)+\sum_{j=p+1}^n \big(\tr(C_j^*C_j\rho_{t-})-1\big)\big(\dd N_j(t)-\dd t\big)\Big),\quad Z_0^\rho=1.\]
Moreover, $(Z_t^\rho)_{t}$ is a nonnegative martingale under $\pp$.
\end{lem}

This martingale property is a crucial ingredient toward the interpretation of $(\rho_t)$ as a quantum measurement process. We refer to the book \cite{BarchielliGregoratti09} and \cite{barchielliholevo} for a complete reference. For any $\rho\in\cD_k$, we define a probability $\pprho\rho_t$ on $(\Omega,\cF_t)$:
\begin{equation} \label{eq_defprho}
	\dd\pprho\rho_t=Z_t^\rho \,\dd\pp|_{\cF_t}.
\end{equation}
Since $(Z_t^\rho)_{t}$ is a $\pp$-martingale from Lemma \ref{lem_Zmg}, the family $(\pprho{\rho}_t)_t$ is consistent, that is $\pprho\rho_t(E)=\pprho\rho_s(E)$ for $t\geq s$ and $E\in\cF_s$. This then defines a unique probability on $(\Omega,\cF_\infty)$, which we denote by $\pprho\rho$. We will denote by $\eerho\rho$ the expectation with respect to $\pprho\rho$. Furthermore note that $\mathbb P^\rho(Z_t^\rho=0)=0$, this way the process $(\rho_t)$ is well defined under $\mathbb P^\rho$ and the arbitrary condition is $\mathbb P^\rho$ almost surely not necessary.

The following proposition makes explicit the relationship between $\pp$ and $\pprho\rho$. It is a classical result in indirect measurement which gives the statistics of the output process.
\begin{prop}\label{prop:girsa}\label{prop:stat} Let $\rho\in\cD_k$. 
For all $i=1,\ldots,p$ and $t\in \R_+$, let
\[\tilde W_i^{\rho}(t)=W_i(t)-\int_0^t\tr\big((L_i+L_i^*)\rho_{s-}\big)\,\dd s.\] Under $\pprho\rho$, the processes $(\tilde W_i^\rho(t)),{i=1,\ldots,p}$ are independent Wiener processes and the processes $(N_j(t)),{j=p+1,\ldots,n}$ are point processes of respective stochastic intensity $\{t\mapsto \tr(C_j^*C_j\rho_{t-})\},{j=p+1,\ldots,n}$.

 In particular under $\mathbb P^\rho$, the processes
\begin{eqnarray}
\tilde W_i^{\rho}(t)&=&W_i(t)-\int_0^t\tr\big((L_i+L_i^*)\rho_{s-}\big)\,\dd s, i=1,\ldots,p\\
\tilde N_j^{\rho}(t)&=&N_j(t)-\int_0^t\tr\big(C_j\rho_{s-}C_j^*\big)\,\dd s, j=p+1,\ldots,n
\end{eqnarray}
are $\mathbb P^\rho$ martingales.
\end{prop}

Now we can express the SDE satisfied by {$(\rho_t)$} in terms of the process $\tilde W^\rho$ and  {consider under $\pprho\rho$ the evolution of a Markov open quantum system subject to indirect measurements.}
 
Using It\^o calculus, one then derives the SDE satisfied by {$(\rho_t)$}: 
\begin{equation}\label{eq:SDE1}
\begin{aligned}
\dd{\rho}_t&=\Lin(\rho_{t-})\dd t\\
&\quad+\sum_{i=1}^p\Big(L_i\rho_{t-}+\rho_{t-} L_i^*-\tr\big(\rho_{t-}(L_i+L_i^*)\big)\rho_{t-}\Big)\dd \tilde W_i^\rho(t)\\
&\quad+\sum_{j=p+1}^n \Big(\frac{C_j\rho_{t-}C_j^*}{\tr(C_j\rho_{t-}C_j^*)}-\rho_{t-}\Big)\dd \tilde N^\rho_j(t)),
\end{aligned}
\end{equation}
where the operator $\mathcal L$ is called Lindbladian operator defined as
$$\mathcal L(\rho)=-i[H,\rho]+\sum_{i=1}^p\left(L_i\rho L_i^*-\frac 12\{L_i^*L_i,\rho\}\right)+\sum_{j=p+1}^n\left(C_j\rho C_j^*-\frac 12\{C_j^*C_j,\rho\}\right),$$
for all $\rho\in\mathcal D_k.$ {Here $\{A,B\}:=AB+BA,$ for all operators $A$ and $B$.}

Remark also that for any $\rho\in \mathcal D_k$, using \eqref{eq_defprho}, we have
\begin{equation} \label{exppsrhos}
\eerho\rho(\rho_t)=\e^{t\Lin}(\rho),
\end{equation}
which is the usual quantum master equation, describing the evolution of a quantum system without measurement.

Now we can introduce the estimated process. {As announced if we do not know the initial state $\rho_0$, one introduces an estimated process starting with an arbitrary initial state $\hat\rho_0$ which evolves conditionally to the output of the measurement as if it was the true trajectory.} More precisely if $\rho_0=\rho$ is the unknown initial state, the natural underlying probability is $\mathbb P^\rho$. This way under $\mathbb P^\rho$, we consider the process 
$$\hat\rho_t=\frac{S_t\hat\rho S_t^*}{\tr(S_t\hat\rho S_t^*)},\quad\hat\rho_0=\hat\rho.$$
Before expressing the SDE satisfied by $(\hat\rho_t)$, we need to impose a necessary condition to be allowed to normalize by the quantity $\tr(S_t\hat\rho S_t^*)$. As for $(\rho_t)$, we can use an arbitrary state in the case where $\tr(S_t\hat\rho S_t^*)=0$. Though in order to compare in an efficient way $(\hat\rho_t)$ and $(\rho_t)$ under $\mathbb P^\rho$, one would expect that $\mathbb P^\rho(\tr(S_t\hat\rho S_t^*)=0)=0$. In general this quantity can vanish under $\mathbb P^\rho$. A usual sufficient condition ensuring this property is $$\ker \hat \rho\subset \ker\rho.$$
This way we have $\mathrm{supp} \rho\subset \mathrm{supp} \hat\rho$, where $\mathrm{supp}$ denotes the support of an operator, that is the orthogonal of the kernel.  This way, there exists a constant $c$ such that
$$\rho\leq c\hat \rho,$$
which implies that $\tr(S_t\rho S_t^*)\leq c \tr(S_t\hat\rho S_t^*)$, for all $t\geq0,$ $\mathbb P^\rho$ almost surely. As a consequence one can see that $$\mathbb P^\rho(\tr(S_t\hat\rho S_t^*)=0)\leq \mathbb P^\rho(\tr(S_t\rho S_t^*)=0)=0$$
and the process $(\hat \rho_t)$ is well defined without requiring an arbitrary state. The process $(\hat\rho_t)$ satisfies the same SDE {as Equation~\eqref{eq:SDE1}} with processes $(\tilde W^{\hat\rho})$ and $(\tilde N^{\hat\rho})$, that is
\begin{equation}
\begin{aligned}
\dd{\hat\rho}_t&=\Lin(\hat\rho_{t-})\dd t\\
&\quad+\sum_{i=1}^p\Big(L_i\hat\rho_{t-}+\hat\rho_{t-} L_i^*-\tr\big(\hat\rho_{t-}(L_i+L_i^*)\big)\hat\rho_{t-}\Big)\dd \tilde W_i^{\hat\rho}(t)\\\
&\quad+\sum_{j=p+1}^n \Big(\frac{C_j\hat\rho_{t-}C_j^*}{\tr(C_j\hat\rho_{t-}C_j^*)}-\rho_{t-}\Big)\dd \tilde N^{\hat\rho}_j(t).\end{aligned}
\end{equation}
In terms of the true signal $(\tilde W^\rho)$ and {$(\tilde N^\rho)$} we get
\begin{equation}
\begin{aligned}
\dd{\hat\rho}_t
&=\Lin(\hat\rho_{t-})\dd t\\
&\quad+\sum_{i=1}^p\Big(L_i\hat\rho_{t-}+\hat\rho_{t-} L_i^*-\tr\big(\hat\rho_{t-}(L_i+L_i^*)\big)\hat\rho_{t-}\Big)\times\\&\hphantom{ccccccccccccccccccccc}\times\left(\dd \tilde W_i^{\rho}(t)+\Big(\tr(\rho_{t-}(L_i+L_i^*)\big)-\tr\big(\hat\rho_{t-}(L_i+L_i^*)\Big)dt\right)\\\
&\quad+\sum_{j=p+1}^n \Big(\frac{C_j\hat\rho_{t-}C_j^*}{\tr(C_j\hat\rho_{t-}C_j^*)}-\rho_{t-}\Big)\left(\dd \tilde N^{\rho}_j(t)+\Big(\tr(C_j\rho_{t-}C_j^*)-\tr(C_j\hat\rho_{t-}C_j^*)\Big)dt\right)
\end{aligned}
\end{equation}
One can then notice that under $\mathbb P^\rho$ the process $(\hat\rho_t)$ is not autonomous and not Markovian. Actually the couple $(\hat\rho_t,\rho_t)$ is Markovian. Furthermore taking the expectation with respect to $\mathbb P^\rho$ we get
\begin{eqnarray}
\dd \mathbb E^\rho[{\hat\rho}_t]&=&\mathbb E^\rho\Bigg[\Lin(\hat\rho_{t-})+\sum_{i=1}^p\Big(L_i\hat\rho_{t}+\hat\rho_{t} L_i^*-\tr\big(\hat\rho_{t}(L_i+L_i^*)\big)\hat\rho_{t}\Big)\times\nonumber\\&&\hphantom{cccccccccccccccccccccccccccccccccc}\times\Big(\tr\big(\rho_{t}(L_i+L_i^*)\big)-\tr\big(\hat\rho_{t}(L_i+L_i^*)\big)\Big)\nonumber\\
&&+\sum_{j=p+1}^n \Big(\frac{C_j\hat\rho_{t}C_j^*}{\tr(C_j\hat\rho_{t}C_j^*)}-\hat\rho_{t}\Big)\Big(\tr(C_j\rho_{t}C_j^*)-\tr(C_j\hat\rho_{t}C_j^*)\Big)\Bigg]dt
\end{eqnarray}
which is not a master equation in Lindblad form.
\bigskip

Now let us introduce a key martingale. Such a martingale is at the cornerstone of the result of \cite{MR4205230}. Here we need to involve $\rho$ and $\hat\rho$.
\begin{prop}\label{mart}
Let define the stochastic process $(M_t^{\rho,\hat\rho})$ by
$$M_t^{\rho,\hat\rho}=\frac{\rho^{1/2}S_t^*S_t\rho^{1/2}}{\tr(S_t^*S_t\hat\rho)}$$
and by an arbitrary state in the case where $\tr(S_t^*S_t\hat\rho)=0$. The process $(M_t^{\rho,\hat\rho})$ is a bounded martingale under $\mathbb P^{\hat\rho}$ which converges in $L^1(\mathbb P^{\hat\rho})$ and $\mathbb P^{\hat\rho}$ almost surely to a random variable $(M_\infty^{\rho,\hat\rho})$.

Furthermore for all $t\geq 0$, we have
 \begin{equation}
 \frac{d\mathbb P^\rho}{d \mathbb P^{\hat\rho}}_{\vert\mathcal F_t}=\tr(M_t^{\rho,\hat\rho}). \end{equation}
Then $\mathbb P^\rho\ll\mathbb P^{\hat\rho}$ and 
 \begin{equation}
\frac{d\mathbb P^\rho}{d \mathbb P^{\hat\rho}}_{\vert\mathcal F_\infty}=\tr(M_\infty^{\rho,\hat\rho})
 \end{equation}
 As a consequence, the process $(M_t^{\rho,\hat\rho})$ converges also in $L^1(\mathbb P^\rho)$ and $\mathbb P^\rho$ almost surely.
\end{prop}
The proof is a slight adaptation of the proof in \cite{MR4205230}. In a sake of completeness, we provide some details in the appendix.

In the next section, {in order to study the limit of the distance between $(\hat\rho_t)$ and $(\rho_t)$ we will use such martingale with respect to the chaotic state.} Let us denote $ch$ when using the chaotic state $\rho=\frac Ik$. In the sequel, we shall use the process $(M_t^{ch,ch})$ defined for all $t$ by $$M_t^{ch,ch}=\frac{S_t^*S_t}{\tr(S_t^*S_t)}.$$
Since $\mathbb P^\rho\ll \mathbb P^{ch}$ for all $\rho\in\mathcal D_k$ and since $\mathbb P^{ch}(\tr(S_t^*S_t)=0)=0$ we have {$\mathbb P^{\rho}(\tr(S_t^*S_t)=0)=0$} for all $\rho\in\mathcal D_k$. This way $M_t^{ch,ch}$ is well defined under $\mathbb P^\rho$ for all $\rho\in\mathcal D_k$. In the sequel, {we will denote $M_t^{ch,ch}=:M_t$ and we will notice that $(M_t)$ converges almost surely and in $L^1$ {under all $\mathbb P^\rho$} for all $\rho\in\mathcal D_k$ (see \cite{MR3947325,MR4205230}).}

\section{Stability of quantum filter, fidelity.}

One of the aspect of the stability of quantum filter concerns the study of a distance between the estimated trajectory $(\hat\rho_t)$ and the true trajectory $(\rho_t)$. A natural distance in this context is the quantum fidelity $F(.,.)$ defined by
$$F(\rho,\sigma)=\tr^2\left(\sqrt{\sqrt{\rho}\sigma\sqrt{\rho}}\right)$$
for all states $(\rho,\sigma)\in(\mathcal D_k)^2$. Recall that the fidelity {is defined} for all density matrices $\rho$, $\sigma,$ all norm one vectors $x,y,$ and for all unitary operators $U$ as follows
\begin{eqnarray*}
F(\rho,\sigma)&=&F(\sigma,\rho),\\
F(U\rho U^*,U\sigma U^*)&=&F(\rho,\sigma),\\
F(\rho,\vert x\rangle\langle x\vert)&=&\tr(\rho\vert x\rangle\langle x\vert)=\langle x\vert\rho\vert x\rangle,\\
F(\vert x\rangle\langle x\vert,\vert y\rangle\langle y \vert)&=&\langle x\vert y\rangle^2,\\
F(\rho,\sigma)=1&\Leftrightarrow&\rho=\sigma.
\end{eqnarray*}

In this section, we shall study the quantum fidelity between the true quantum trajectory $(\rho_t)$ and its estimation $(\hat\rho_t)$. So far, in this direction, an interesting result was the result showing that the fidelity between the estimated trajectory and the true trajectory is a submartingale \cite{MR2884021,MR3255126}. Here we make precise the fidelity in terms of $(M_t)$ which allows to study its behavior in long time. As a particular case, under the following purification assumption we show that the fidelity converges to $1$, {i.e., the stability is ensured}.
\bigskip

\noindent\textbf{Purification Assumption:} Any orthogonal projection $\pi$ such that for all $i=1,\ldots,p$ and for all $j=p+1,\ldots,n$, there exists $\lambda_k,k=1,\ldots,n$ such that
$$\pi(L_i+L_i^*)\pi=\lambda_i\pi,\quad \textrm{and}\quad \pi C_j^*C_j\pi=\lambda_j\pi$$
is of rank one.
\bigskip

 In \cite{MR4205230}, it has been shown that this condition implies the convergence of $(M_t)$ toward a rank one orthogonal projector.

 Here we have the following result.
\begin{prop}
Let $(\rho_t)$ { and $(\hat\rho_t)$ be the true quantum trajectory and its estimated one starting respectively at $\rho_0=\rho$ and $\hat\rho_0=\hat\rho.$ Suppose that $\ker\hat\rho\subset \ker\rho$.}  The fidelity between the true quantum trajectory and its estimation $(\hat\rho_t)$ satisfy for all $t\geq0$
\begin{equation}\label{eq:fidelity}F(\hat \rho_t,\rho_t)=F\left(\frac{S_t\rho S_t^*}{\tr(S_t\rho S_t)},\frac{S_t\hat\rho S_t^*}{\tr(S_t\hat\rho S_t)}\right)=F\left(\frac{\sqrt{M_t}\rho \sqrt{M_t}}{\tr(M_t\rho )},\frac{\sqrt{M_t}\hat\rho \sqrt{M_t}}{\tr(M_t\hat\rho )}\right).\end{equation}
As a consequence $(F(\hat \rho_t,\rho_t))$ converges $\mathbb P^\rho$ almost surely and in $L^1(\mathbb P^\rho)$ toward
$$F\left(\frac{\sqrt{M_\infty}\rho \sqrt{M_\infty}}{\tr(M_\infty\rho )},\frac{\sqrt{M_\infty}\hat\rho \sqrt{M_\infty}}{\tr(M_\infty\hat\rho )}\right).$$
As a consequence under the Purification assumption, we have that
$$\lim_{t\rightarrow\infty}F(\rho_t,\hat\rho_t)=1$$
in $L^1(\mathbb P^\rho)$ and $\mathbb P^\rho$ almost surely.
\end{prop}

\begin{proof}
In order to prove the equality \eqref{eq:fidelity}, we use the polar decomposition of $S_t$, that is
$$S_t=U_t\sqrt{S_t^*S_t}=U_t\sqrt{\tr(S_t^*S_t)}\sqrt {M_t}$$ 
This way for all $\rho$ and for all $t\geq0$ we have
$$\rho_t=\frac{S_t\rho S_t^*}{\tr(S_t\rho S_t)}=\frac{U_t\sqrt{M_t}\rho \sqrt{M_t}U_t^*}{\tr(M_t\rho )}$$
then
\begin{eqnarray}F(\hat \rho_t,\rho_t)&=&F\left(\frac{S_t\rho S_t^*}{\tr(S_t\rho S_t)},\frac{S_t\hat\rho S_t^*}{\tr(S_t\hat\rho S_t)}\right)\\&=&F\left(\frac{U_t\sqrt{M_t}\rho \sqrt{M_t}U_t^*}{\tr(M_t\rho )},\frac{U_t\sqrt{M_t}\hat\rho \sqrt{M_t}U_t^*}{\tr(M_t\hat\rho )}\right)\\
&=&F\left(\frac{\sqrt{M_t}\rho \sqrt{M_t}}{\tr(M_t\rho )},\frac{\sqrt{M_t}\hat\rho \sqrt{M_t}}{\tr(M_t\hat\rho )}\right)\end{eqnarray}
by the unitary invariance. The convergence result is then a consequence of the convergence of $(M_t)$ under all probabilities $\mathbb P^\rho$ and the fact that the fidelity is {continuous}. The only thing that we need to take care is the possibility that $\tr(M_\infty\hat\rho )=0$ but since $\tr(M_\infty\rho )\leq c\tr(M_\infty\hat\rho )$ this would imply that $\tr(M_\infty\rho )=0$ which appears with $\mathbb P^\rho$ probability zero.

Now under purification, we know (see \cite{MR4205230}) that there exists a random variable $z$ of norm one such that 
$$M_\infty=\vert z\rangle\langle z\vert,$$
$\mathbb P^\rho$ almost surely. This way we have 
\begin{eqnarray}
F\left(\frac{\sqrt{M_\infty}\rho \sqrt{M_\infty}}{\tr(M_\infty\rho )},\frac{\sqrt{M_\infty}\hat\rho \sqrt{M_\infty}}{\tr(M_\infty\hat\rho )}\right)&=&
F\left(\frac{\sqrt{\vert z\rangle\langle z\vert}\rho \sqrt{\vert z\rangle\langle z\vert}}{\tr(\vert z\rangle\langle z\vert\rho )},\frac{\sqrt{\vert z\rangle\langle z\vert}\hat\rho \sqrt{\vert z\rangle\langle z\vert}}{\tr(\vert z\rangle\langle z\vert\hat\rho )}\right)\nonumber\\
&=&F(\vert z\rangle\langle z\vert,\vert z\rangle\langle z\vert)=1
\end{eqnarray}
and the result is proven.
\end{proof}

The almost sure convergence of the fidelity was already known since it has {been} proven in \cite{MR2884021} for a discrete time version and in \cite{MR3255126} for a continuous time version that the fidelity is a submartingale. Since it is nonnegative and bounded it almost surely converges. Here we express the limit in terms of the limit $M_\infty$ which completes the works \cite{MR2884021,MR3255126}. In the purification context, we show that the quantum filter is stable. In the general case, we have that
$$F\left(\frac{\sqrt{M_\infty}\rho \sqrt{M_\infty}}{\tr(M_\infty\rho )},\frac{\sqrt{M_\infty}\hat\rho \sqrt{M_\infty}}{\tr(M_\infty\hat\rho )}\right)=1\Leftrightarrow\frac{\sqrt{M_\infty}\rho \sqrt{M_\infty}}{\tr(M_\infty\rho )}=\frac{\sqrt{M_\infty}\hat\rho \sqrt{M_\infty}}{\tr(M_\infty\hat\rho )}$$  
In particular if $\pi_\infty$ denotes the projector onto the support of $M_\infty$ we have
$$\frac{\sqrt{M_\infty}\rho \sqrt{M_\infty}}{\tr(M_\infty\rho )}=\frac{\sqrt{M_\infty}\hat\rho \sqrt{M_\infty}}{\tr(M_\infty\hat\rho )}\Leftrightarrow \exists \lambda_\infty\,\,s.t\,\,\pi_\infty \hat\rho\pi_\infty=\lambda_\infty\pi_\infty \rho\pi_\infty$$
This way, one can see that obtaining the stability of the filter needs a strong assumption on the initialization of the filter. Indeed the estimated starting state needs to be proportional to $\rho$ onto the support on $M_\infty$ which is clearly satisfied if $M_\infty$ is {a rank one} orthogonal projector. In general, if $\textrm{rank}(M_\infty)>1$, it is not possible unless choosing $\hat\rho=\rho$ which is of course not allowed. In general a natural choice for the starting state of $(\hat\rho_t)$ is the chaotic state that is $\hat\rho_0=\frac Ik$ in order to ensure the property $\ker\hat\rho\subset \ker\rho$.

\begin{rmq}We note that  for a qubit, purification can be completely described. As in dimension two, the only orthogonal projector of rank two is the identity operator. This way, either purification holds or all $L_i$ satisfy $L_i+L_i^*=\alpha_i I$ for some constant $\alpha_i$ and all $C_j$ are multiple of unitary operators. Let $i\in\{1,\ldots,p\}$, the condition $L_i+L_i^*=\alpha_i I$ means that $L_i=A_i+\beta_iI$ for some anti-hermitian operator $A_i$ and some constant $\beta_i$.
\end{rmq}

In the case where the stability is not satisfied, one can wonder if it is the case for the Cesaro means. Such consideration follows the line of the ergodic Theorem of K\"ummerer-Maassen \cite{HMergo}. This is the question addressed in the next section. Note that in the case that the stability holds the answer is trivial.

\section{{Cesaro mean} of the estimated filter}

Concerning large time behavior of quantum trajectories, a first major result is the pathwise ergodic theorem of Maassen. This theorem essentially shows that the Cesaro  mean of quantum trajectories converges toward an invariant state of the Lindblad master equation. We denote by $$\mathcal I=\{\rho\in\mathcal D_k\,\vert \,\mathcal L(\rho)=0\}$$
that is the set of equilibrium states. Now we state the ergodic theorem for quantum trajectories.

\begin{thm}[Ergodic Theorem of K\"ummerer-Maassen \cite{HMergo}]\label{thm:ergo}
Let $(\rho_t)$ be the quantum trajectory starting with initial condition $\rho_0=\rho$. Then {there exists} a random variable $\Theta^\rho_\infty$ valued in $\mathcal I$ such that
{$$\lim_{t\rightarrow\infty}\frac{1}{t}\int_0^t\rho_sds=\Theta^\rho_\infty$$}
$\mathbb P^\rho$ almost surely.
\end{thm}
The proof of this result relies on martingale property and on the fact that $\mathbb E^\rho[\rho_t]=e^{t\mathcal L}(\rho)$. In this section, we want to address the same question for the estimated trajectory $(\hat\rho_t)$ that is what is the limit under $\mathbb P^{\rho}$ of 
$$\frac{1}{t}\int_0^t\hat\rho_sds.$$
It is not straightforward to adopt the approach of the proof of the original ergodic theorem since the mean evolution of $(\hat\rho_t)$ under $\mathbb P^{\hat\rho}$ is not in Lindblad form and even more not linear. {Hence a lot of key points cannot be} directly obtained. Nevertheless, the absolute continuity $\mathbb P^\rho\ll \mathbb P^{\hat\rho}$ makes the proof of the following result very simple. Such approach with absolute continuity has already been used successfully in the context of quantum non demolition in \cite{bbbqnd} where they prove that an estimated filter has the same behavior as the true quantum trajectory.

\begin{thm}\label{thm:ergoest}
{Let $(\rho_t)$ and $(\hat\rho_t)$ be the quantum trajectory and the estimated one starting respectively with initial conditions $\rho_0=\rho$ and $\hat\rho_0=\hat\rho$ such that $\ker\hat\rho\subset\ker\rho$.} Then, there exists a random variable $\Theta^{\hat\rho}_\infty$ valued in $\mathcal I$ such that
$$\lim_{t\rightarrow\infty}\frac{1}{t}\int_0^t\hat\rho_sds=\Theta^{\hat\rho}_\infty$$
$\mathbb P^\rho$ almost surely.
\end{thm}

\begin{proof}
As mentioned the proof of this result is an application of the absolute continuity result. Indeed since $\mathbb P^\rho\ll\mathbb P^{\hat\rho}$, an almost sure result for $\mathbb P^{\hat\rho}$ will imply an almost sure result for $\mathbb P^\rho$. As a consequence, using the ergodic Theorem \ref{thm:ergo} we have that
$$\mathbb P^{\hat\rho}\left(\lim_{t\rightarrow\infty}\frac{1}{t}\int_0^t\hat\rho_sds=\Theta^{\hat\rho}_\infty\right)=1$$
then 
$$\mathbb P^{\rho}\left(\lim_{t\rightarrow\infty}\frac{1}{t}\int_0^t\hat\rho_sds=\Theta^{\hat\rho}_\infty\right)=1,$$
which is the desired result.
\end{proof}

\begin{cor}
{Let $(\rho_t)$ and $(\hat\rho_t)$ be the quantum trajectory and the estimated one starting respectively with initial conditions $\rho_0=\rho$ and $\hat\rho_0=\hat\rho$ such that $\ker\hat\rho\subset\ker\rho$.} If $\mathcal I$ is reduced to one element $\rho_\infty$, that is there exists a unique invariant state for the Lindblad master equation then 
$$\lim_{t\rightarrow\infty}\frac{1}{t}\int_0^t\rho_sds=\lim_{t\rightarrow\infty}\frac{1}{t}\int_0^t\hat\rho_sds=\rho_{\infty}$$
$\mathbb P^{\rho}$ almost surely.
\end{cor}
Note that Theorem \ref{thm:ergoest} tells that the Cesaro mean of the estimated trajectory converges toward a random invariant state under $\mathbb P^\rho$ as well as the true trajectory. In general, when there are more than one invariant state, there are an infinity of {states} (every convex combination). Then there {is} no reason that the invariant state obtained for the Cesaro mean of the estimated trajectory is the same as { the true one}. Indeed the random variable $\Theta^{\hat\rho}_\infty$ has not the same law under $\mathbb P^\rho$ and $\mathbb P^{\hat\rho}$. One can easily construct examples where the two {Cesaro means} have different limits (see the end of this section).

In the sequel of this section, we shall investigate a natural situation where there are more than one invariant state and where we are able to show that both Cesaro means converge to the same invariant state. 

We shall adopt the notation of \cite{BN} where the set of {invariant states} is described in terms of enclosure. We refer also to \cite{CarbonePautrat,MR3559660} where similar results are exposed.
We consider
{$$\mathcal D=\{\psi\in\mathcal H\,\,s.t\,\,\langle\psi,e^{t\mathcal L}(\rho)\psi\rangle\xrightarrow[t\rightarrow\infty]{}0\,\,\forall \rho\in\mathcal D_k\}\quad \textrm{and}\quad\mathcal R=\mathcal D^\perp.$$}
In particular $\mathcal R$ contains the support of all invariant states. We consider the situation where 
$$\mathcal D=0,\quad\mathcal R=\bigoplus_{i=1}^K\mathcal V_i$$
where each $\mathcal V_i$ is the support of a minimal invariant state and where this decomposition is unique. We denote by $\rho_\infty^i$ the invariant state with support equals to $\mathcal V_i$. The minimality is understood into the fact that there is no invariant state whose support is strictly included in $\mathcal V_i$. For each $\rho_\infty^i$, we consider the positive associated operator $M_i$ such that $\mathcal L^*(M_i)=0$, $\tr(M_k\rho^i_\infty)=\delta_{ki}$. The operators $M_i$ are orthogonal projectors and {satisfy $\sum_{i=1}^K M_i=Id$}. Note that the case $\mathcal D\neq0$ is more tricky and will be considered in full generality in \cite{BBP2020} where a law of large number, a central limit theorem as well as a large deviation principle will be derived. The article \cite{BBP2020} is not linked with estimated filters and will investigate discrete time quantum trajectories. 

Now let us introduce
$$Q^{\rho}_i(t)=\tr(M_i\rho_t),$$
for all $t\geq0$ and for all $i=1,\ldots,K$. This process can be interpreted as the probability of "seing" the quantum trajectory $(\rho_t)$ into the support of $\rho_{\infty}^i$. It is as if we measure directly an observable of the form $A=\sum_{i=1}^K\mu_iM_i$ on the state $(\rho_t)$ at time $t$. 

We shall need the following assumption in order to derive the following results.
\bigskip

\noindent\textbf{Identifiability Assumption:} For all $u\neq v$, $u,v=1,\ldots,K,$ there exists $i=1,\ldots,p$ or $j=p+1,\ldots,n$ such that
$$\tr((L_i+L_i^*)\rho_\infty^u)\neq\tr((L_i+L_i^*)\rho_\infty ^v)\quad \textrm{or}\quad\tr(C_j\rho_\infty ^uC_j^*)\neq\tr(C_j\rho_\infty^vC_j^*).$$
\bigskip

\noindent\textbf{Spectral Assumption:} The operator $\mathcal L$ has no purely imaginary eigenvalues.
\bigskip

The proof of the following result is given in the appendix.

\begin{prop}\label{mart2}
Let $\rho\in\mathcal D_k$. The processes $(Q^\rho_i(t)), i=1,\ldots,K$ are bounded martingale under $\mathbb P^\rho$. They converge in $L^1(\mathbb P^\rho)$ and $\mathbb P^\rho$ almost surely toward random variables $Q^\rho_i(\infty),i=1,\ldots,K$.

Assume that the Identifiability and the Spectral Assumption hold, then $$Q^\rho_v(\infty)Q^\rho_u(\infty)=0$$
$\mathbb P^\rho$ almost surely. Then $Q^\rho_u(\infty)\in\{0,1\}$.

 In particular there exists a random variable $\Gamma$ valued in $\{1,\ldots,K\}$ such that $Q^{\rho}_{\Gamma}(\infty)=1$, $\mathbb P^\rho$ almost surely  and $$\mathbb P^\rho[\Gamma=i]=\mathbb P^\rho(Q_i(\infty)=1)=Q^\rho_i(0).$$
\end{prop}

\begin{thm}\label{thm:Cesarogen}Let $(\rho_t)$ be the quantum trajectory started with initial condition $\rho_0=\rho$ and $(\hat\rho_t)$ the estimated trajectory starting with $\hat\rho_0=\hat\rho$ such that $\ker\hat\rho\subset\ker\rho$. Assume that the Identifiability and the Spectral Assumptions hold, then
$$\lim_{t\rightarrow\infty}\frac{1}{t}\int_0^t\rho_sds=\lim_{t\rightarrow\infty}\frac{1}{t}\int_0^t\hat\rho_sds=\Theta^\rho_{\infty}=\sum_{i=1}^KQ_i^\rho(\infty)\rho^i_\infty$$
$\mathbb P^\rho$ almost surely.
\end{thm}

\begin{proof}
{The proof is again based on application of the absolute continuity.} Since $\ker\hat\rho\subset\ker\rho$, note that
$$Q_u^\rho(t)\leq cQ_u^{\hat\rho}(t)$$
for all $t\geq 0$ and for all $u=1,\ldots,K$. This way we have $Q_u^\rho(\infty)\leq cQ_u^{\hat\rho}(\infty)$. Under $\mathbb P^{\hat\rho}$ or $\mathbb P^\rho$, only one $Q_u^\rho(\infty)$ is non zero and only one $Q_v^{\hat\rho}(\infty)$ is non zero. Necessarily $u=v$. Indeed let $u$ be such that $Q_u^\rho(\infty)=1$ then $Q_u^{\hat\rho}(\infty)>0$ and then $Q_u^{\hat\rho}(\infty)=1$. Reciprocally let $u$ be such that $Q_u^{\hat\rho}(\infty)=1$ then for all $v\neq u$, we have $Q_v^{\hat\rho}(\infty)=0$ which implies $Q_v^{\rho}(\infty)=0$ and then $Q_u^{\rho}(\infty)=1$.  We finally have the following equivalence
$$Q_v^\rho(\infty)=0\Leftrightarrow Q_v^{\hat\rho}(\infty)=0\quad \mathrm{and}\quad Q_u^\rho(\infty)=1\Leftrightarrow Q_u^{\hat\rho}(\infty)=1.$$
Now we know that under $\mathbb P^\rho$ (see the proof in the appendix), we have
$$\Theta^\rho_\infty=\sum_{i=1}^K Q^\rho_i(\infty)\rho_\infty^i\quad
\mathrm{and} 
\quad\Theta^{\hat\rho}_\infty=\sum_{i=1}^K Q^{\hat\rho}_i(\infty)\rho_\infty^i.$$
The above equivalence yields the result.
\end{proof}

Theorem \ref{thm:Cesarogen} clearly states that under the Identifiability Assumption and Spectral Assumption the Cesaro mean of the estimated filter converges toward the same limit of the one for the true quantum trajectory. The identifiability assumption is clearly the good setup for such a result. Indeed consider simply block type operators of the form
$$L_i=\left(\begin{array}{cc}\tilde L_i&0\\0&\tilde L_i\end{array}\right)\quad C_j=\left(\begin{array}{cc}\tilde C_j&0\\0&\tilde C_j\end{array}\right)$$
and suppose that
$$\rho_0=\left(\begin{array}{cc} 0&0\\0&\tilde \rho\end{array}\right)$$
Suppose for example that the operators $\tilde L_i$ and $\tilde C_j$ generate a unique invariant state $\tilde\rho_\infty$ for the corresponding Lindblad operator. It is then clear that {if the estimated filter is initialized at $\hat\rho_0=I/k$, then} the Cesaro mean of the estimated filter will converge to
$$\left(\begin{array}{cc} \frac 12\tilde\rho_\infty&0\\0&\frac 12\tilde\rho_\infty\end{array}\right)$$
whereas {the Cesaro mean} of the true quantum trajectory converges toward 
$$\left(\begin{array}{cc} 0&0\\0&\tilde\rho_\infty\end{array}\right).$$
Of course the Identifiability Assumption is obviously not satisfied.

Concerning the Spectral Assumption, we suspect that it is too strong, however it is needed in the proof of a version of the Ergodic Theorem of Kummerer Maasen for a subsequence (see the appendix). In particular, since we do not have the control on this subsequence, the spectral condition allows us to conclude.
\bigskip

\noindent\textbf{Acknowledgments.} C. P. is supported by the ANR project ``Quantum Trajectories" ANR-20-CE40-0024-01, 2021-2025. C.P. is supported by the ANR project ``\'Evolutions Stochastiques Quantiques" ESQuisses Projet-ANR-20-CE47-0014. N. A. and C. P. are supported by the ANR project ``Estimation et contr\^ole des systèmes quantiques ouverts" Q-COAST Projet-ANR-19-CE48-0003. N. A. is supported by the ANR project QUACO ANR-17-CE40-0007.

\bibliographystyle{halpha}
\bibliography{biblio}
\section{Appendix}
\label{sec:app}

\begin{proof}[Proof of Proposition \ref{mart}]
Introduce the process $R_t=\frac{S_t\rho^{1/2}}{\sqrt{\tr(S_t^*S_t\hat \rho)}}$ for all $t\geq0$, then applying It\^o calculus yields
\begin{align*}
\dd M_t^{\rho,\hat\rho}=&\sum_{i=1}^p \Big(R_{t-}^*(L_i+L_i^*)R_{t-}-M_{t-}^{\rho,\hat\rho}\tr\big(R_{t-}^*(L_i+L_i^*)R_{t-}\big)\Big)\,\dd \tilde W^{\hat\rho}_i(t)\\
&+\sum_{j=p+1}^n \Big(\frac{R_{t-}^*C_j^*C_j R_{t-}}{\tr(R_{t-}^*C_j^*C_j R_{t-})}- M_{t-}^{\rho,\hat\rho}\Big)\dd \tilde N^{\hat\rho}_j(t).
\end{align*}
Then Proposition \ref{prop:stat} yields the result concerning the martingale property under $\mathbb P^{\hat\rho}$. Now note that $(M_t^{\rho,\hat\rho})$ is $\mathbb P^{\hat\rho}$ almost surely a nonnegative operator. Since $\rho\leq c\hat\rho$, we have $$\tr(M_t^{\rho,\hat\rho})\leq c,$$
for all $t\geq0$, $\mathbb P^{\hat\rho}$-almost surely. Then under $\mathbb P^{\hat\rho}$, $(M_t^{\rho,\hat\rho})$ is a nonnegative bounded martingale which converges in $L^1(\mathbb P^{\hat\rho})$ and $\mathbb P^{\hat\rho}$ almost surely to a random variable $M_\infty^{\rho,\hat\rho}$.

Now the following equality
\begin{eqnarray}
\tr(S_t\rho S_t^*)&=&\frac {\tr(S_t\rho S_t^*)}{\tr(S_t\hat\rho S_t^*)}\tr(S_t\hat\rho S_t^*)
\end{eqnarray}
for all $t\geq0$ implies that $$\frac{d\mathbb P^{\rho}}{d\mathbb P^{\hat\rho}}_{\vert\mathcal F_t}=\tr(M_t^{\rho,\hat\rho})$$
and the $L^1(\mathbb P^{\hat\rho})$ convergence yields the absolute continuity result that is
$$\frac{d\mathbb P^{\rho}}{d\mathbb P^{\hat\rho}}=\tr(M_\infty^{\rho,\hat\rho}).$$
\end{proof}

In the proof of Proposition \eqref{mart2}, we shall need the intermediate ergodic result which actually is not a direct consequence of the Ergodic Theorem of Kummerer Maassen.

\begin{prop}\label{prop:ergogener}
Let $(\rho_t)$ be a quantum trajectory starting with initial condition $\rho_0=\rho$. Assume that $\mathcal L$ satisfies the Spectral Assumption. For all increasing subsequence $(t_n)$ converging toward $\infty$, there exists a random variable $\Theta_\infty$ (which depends possibly of the subsequence) such that
\begin{equation}
\lim_{n\rightarrow}\frac 1n\sum_{k=1}^n\rho_{t_k}=\Theta_\infty
\end{equation}
$\mathbb P^\rho$ almost surely.
\end{prop}

\begin{proof}
This result is very close to the ergodic theorem of Kummerer Maassen and the proof is almost exactly the same. Actually the only point to check is that the process $(P(\rho_{t_k}))$ is a bounded convergent martingale toward $\Theta_\infty$, where
\begin{equation}\label{eq:invergo}P(\rho)=\lim_{n\rightarrow\infty}\frac 1n\sum_{k=1}^ne^{t_k\mathcal L}(\rho).
\end{equation}
Since $\mathcal L$ has no purely imaginary part eigenvalues $(e^{t_k\mathcal L}(\rho))$ is a convergent sequence toward an invariant state. Then $P(\rho)$ is an invariant state. Now we easily have
\begin{eqnarray}
\mathbb E^\rho[P(\rho_{t_j})\vert\mathcal F_{j-1}]&=& \lim_{n\rightarrow\infty}\frac 1n\sum_{k=1}^ne^{t_k\mathcal L}(\mathbb E^\rho[(\rho_{t_j})\vert\mathcal F_{j-1}])\\
&=&\lim_{n\rightarrow\infty}\frac 1n\sum_{k=1}^ne^{t_k\mathcal L}(e^{(t_j-t_{j-1})\mathcal L}(\rho_{t_{j-1}}))\\
&=&e^{(t_j-t_{j-1})\mathcal L}(P(\rho_{t_{j-1}}))=P(\rho_{t_{j-1}})
\end{eqnarray}
Now the interesting reader can follow the line of \cite{HMergo} to complete the proof. We just want to point out that since we do not have the hand of the subsequence, we need to impose the condition $\mathcal L$ has no purely imaginary part in order that $P(\rho)$ is an invariant state. In the case where $\mathcal L$ has purely imaginary eigenvalues the convergence, of the Cesaro mean \eqref{eq:invergo} toward an invariant state is not ensured (consider unitary operator for example).
\end{proof}

\begin{proof}[Proof of Proposition \eqref{mart2}]
For $k=1,\ldots,K$, the fact that $(Q_k^\rho(t))$ is a martingale under $\mathbb P^{\rho}$ is straightforward. Using the fact that $\mathcal L^*(M_k)=0$, we indeed get
\begin{eqnarray}
dQ_k^\rho(t)&=&\sum_{i=1}^p\Big(\tr(M_k(L_i\rho_{t-}+\rho_{t-} L_i^*)-\tr\big(\rho_{t-}(L_i+L_i^*)\big)\tr(M_k\rho_{t-}))\Big)\dd \tilde W_i^\rho(t)\nonumber\\
&&+\sum_{j=p+1}^n \left(\frac{\tr(M_kC_j\rho_{t-}C_j^*)}{\tr(C_j\rho_{t-}C_j^*)}-\tr(M_k\rho_{t-})\right)\dd \tilde N^\rho_j(t)),
\end{eqnarray}
which is clearly a martingale under $\mathbb P^\rho$. Now, compute the second moment, for all $t\geq0$ we have
\begin{eqnarray}
\mathbb E^\rho[(Q_k^\rho(t))^2]&=&\mathbb E^\rho\Bigg[ \int_0^t\sum_{i=1}^p\Big(\tr(M_k(L_i\rho_{s}+\rho_{s} L_i^*)-\tr\big(\rho_{s}(L_i+L_i^*)\big)\tr(M_k\rho_{s}))\Big)^2ds\nonumber\\
&&+\int_0^t\sum_{j=p+1}^n\left(\frac{\tr(M_kC_j\rho_{s}C_j^*)}{\tr(C_j\rho_{s}C_j^*)}-\tr(M_k\rho_{s})\right)^2\tr(C_j\rho_{s}C_j^*)ds\Bigg].
\end{eqnarray} 
The fact that $(Q_k^\rho(t))$ is bounded and $\mathbb P^\rho$ almost surely implies convergence of all the moment. This quantity converges namely when $n$ goes to infinity toward
\begin{eqnarray}
&& \int_0^\infty\mathbb E^\rho\Bigg[\sum_{i=1}^p\Big(\tr(M_k(L_i\rho_{s}+\rho_{s} L_i^*)-\tr\big(\rho_{s}(L_i+L_i^*)\big)\tr(M_k\rho_{s}))\Big)^2\Bigg]ds\nonumber\\
&&+\int_0^\infty\mathbb E^\rho\Bigg[\sum_{j=p+1}^n\left(\frac{\tr(M_kC_j\rho_{s}C_j^*)}{\tr(C_j\rho_{s}C_j^*)}-\tr(M_k\rho_{s})\right)^2\tr(C_j\rho_{s}C_j^*)\Bigg]ds.
\end{eqnarray}
Since the integrand are nonnegative, their liminf at infinity is equal to zero. Then there exists a increasing subsequence $(t_n)$ converging toward infinity such that
\begin{eqnarray}
&& \lim_{n\rightarrow\infty}\mathbb E^\rho\Bigg[\sum_{i=1}^p\Big(\tr(M_k(L_i\rho_{t_n}+\rho_{t_n} L_i^*)-\tr\big(\rho_{t_n}(L_i+L_i^*)\big)\tr(M_k\rho_{t_n}))\Big)^2\nonumber\\
&&\hphantom{cccccccc}+\sum_{j=p+1}^n\left(\frac{\tr(M_kC_j\rho_{t_n}C_j^*)}{\tr(C_j\rho_{t_n}C_j^*)}-\tr(M_k\rho_{t_n})\right)^2\tr(C_j\rho_{t_n}C_j^*)\Bigg]=0.
\end{eqnarray}
Then, the sequence
\begin{eqnarray}&&\sum_{i=1}^p\Big(\tr(M_k(L_i\rho_{t_n}+\rho_{t_n} L_i^*)-\tr\big(\rho_{t_n}(L_i+L_i^*)\big)\tr(M_k\rho_{t_n}))\Big)^2
\\&&+\sum_{j=p+1}^n\left(\frac{\tr(M_kC_j\rho_{t_n}C_j^*)}{\tr(C_j\rho_{t_n}C_j^*)}-\tr(M_k\rho_{t_n})\right)^2\tr(C_j\rho_{t_n}C_j^*)
\end{eqnarray}
 converges in $L^1(\mathbb P^\rho)$ and then up to a subsequence we have a convergence $\mathbb P^\rho$ almost surely. Up to extracting again, there exists an increasing subsequence still denoted by $(t_n)$ converging toward infinity such that
\begin{eqnarray}
&& \lim_{n\rightarrow\infty}\sum_{i=1}^p\Big(\tr(M_k(L_i\rho_{t_n}+\rho_{t_n} L_i^*)-\tr\big(\rho_{t_n}(L_i+L_i^*)\big)\tr(M_k\rho_{t_n}))\Big)^2\nonumber\\
&&\hphantom{cccccccc}+\sum_{j=p+1}^n\Big(\frac{\tr(M_kC_j\rho_{t_n}C_j^*)}{\tr(C_j\rho_{t_n}C_j^*)}-\tr(M_k\rho_{t_n})\Big)^2\tr(C_j\rho_{t_n}C_j^*)=0.
\end{eqnarray}
Invoking $Q^\rho_k(\infty)$, we get
\begin{eqnarray}
&& \lim_{n\rightarrow\infty}\sum_{i=1}^p\Big(\tr(M_k(L_i\rho_{t_n}+\rho_{t_n} L_i^*)-\tr\big(\rho_{t_n}(L_i+L_i^*)\big)Q^\rho_k(\infty)\Big)^2\nonumber\\
&&\hphantom{cccccccc}+\sum_{j=p+1}^n\left(\frac{\tr(M_kC_j\rho_{t_n}C_j^*)}{\tr(C_j\rho_{t_n}C_j^*)}-Q^\rho_k(\infty)\right)^2\tr(C_j\rho_{t_n}C_j^*)=0,
\end{eqnarray}
$\mathbb P^\rho$ almost surely. This way, for all $i=1,\ldots,p$ and $j=p+1,\ldots,n$
\begin{eqnarray}
&& \lim_{n\rightarrow\infty}\Big(\tr(M_k(L_i\rho_{t_n}+\rho_{t_n} L_i^*)-\tr\big(\rho_{t_n}(L_i+L_i^*)\big)Q^\rho_k(\infty)\Big)^2=0\nonumber\\
&&\lim_{n\rightarrow\infty}\left(\frac{\tr(M_kC_j\rho_{t_n}C_j^*)}{\tr(C_j\rho_{t_n}C_j^*)}-Q^\rho_k(\infty)\right)^2\tr(C_j\rho_{t_n}C_j^*)=0,
\end{eqnarray}
which yields by multiplying by $\tr(C_j\rho_{t_n}C_j^*)$ which is bounded
\begin{eqnarray}
&& \lim_{n\rightarrow\infty}\Big(\tr(M_k(L_i\rho_{t_n}+\rho_{t_n} L_i^*)-\tr\big(\rho_{t_n}(L_i+L_i^*)\big)Q^\rho_k(\infty)\Big)^2=0\nonumber\\
&&\lim_{n\rightarrow\infty}\Big(\frac{\tr(M_kC_j\rho_{t_n}C_j^*)}{\tr(C_j\rho_{t_n}C_j^*)}-Q^\rho_k(\infty)\Big)^2\tr(C_j\rho_{t_n}C_j^*)^2\\&&\hphantom{ccccc}=\lim_{n\rightarrow\infty}\Big(\tr(M_kC_j\rho_{t_n}C_j^*)-\tr(C_j\rho_{t_n}C_j^*)Q^\rho_k(\infty)\Big)^2=0.
\end{eqnarray}
Finally, we get
\begin{eqnarray}
&& \lim_{n\rightarrow\infty}\Big(\tr(M_k(L_i\rho_{t_n}+\rho_{t_n} L_i^*)-\tr\big(\rho_{t_n}(L_i+L_i^*)\big)Q^\rho_k(\infty)\Big)=0\nonumber\\
&&\lim_{n\rightarrow\infty}\Big(\tr(M_kC_j\rho_{t_n}C_j^*)-\tr(C_j\rho_{t_n}C_j^*)Q_k(\infty\Big)=0.
\end{eqnarray}
Using the version of the ergodic result, Proposition \ref{prop:ergogener} we get
\begin{eqnarray}
&& \lim_{n\rightarrow\infty}\frac{1}{n}\sum_{l=1}^n\Big(\tr(M_k(L_i\rho_{t_l}+\rho_{t_l} L_i^*)-\tr\big(\rho_{t_l}(L_i+L_i^*)\big)Q^\rho_k(\infty)\Big)\\&=&\Big(\tr(M_k(L_i\Theta_{\infty}+\Theta_{\infty} L_i^*)-\tr\big(\Theta_{\infty}(L_i+L_i^*)\big)Q^\rho_k(\infty)\Big)=0\nonumber\\
&&\lim_{n\rightarrow\infty}\frac{1}{n}\sum_{l=1}^n\Big(\tr(M_kC_j\rho_{t_l}C_j^*)-\tr(C_j\rho_{t_l}C_j^*)Q^\rho_k(\infty)\Big)\nonumber\\
&=&\Big(\tr(M_kC_j\Theta_{\infty}C_j^*)-\tr(C_j\Theta_{\infty}C_j^*)Q^\rho_k(\infty)\Big)=0
\end{eqnarray}
Now let us note that
$$\tr(M_k\Theta_\infty)=\lim\frac{1}{n}\sum_{j=1}^n\tr(M_u\rho_{t_j})=\lim\frac{1}{n}\sum_{j=1}^nQ_u(t_j)=Q_k(\infty)$$
and we then have
{$$\Theta_\infty=\sum_{u=1}^K Q_u(\infty)\rho_\infty^u.$$}
Here, we can note that the limit does not depend of the subsequence. Let us exploit that for all $i=1,\ldots,p$

\begin{eqnarray*}
\Big(\tr(M_k(L_i\Theta_{\infty}+\Theta_{\infty} L_i^*)&-&\tr\big(\Theta_{\infty}(L_i+L_i^*)\big)Q^\rho_k(\infty)\Big)=0\\
\Rightarrow\quad\tr(M_k(L_i\Theta_{\infty}+\Theta_{\infty} L_i^*)&=&\tr\big(\Theta_{\infty}(L_i+L_i^*)\big)Q^\rho_k(\infty).
\end{eqnarray*}
Multiplying by $Q_v^\rho(\infty)$, we get
\begin{eqnarray}Q^\rho_v(\infty)\tr(M_k(L_i\Theta_{\infty}+\Theta_{\infty} L_i^*))&=&\tr\big(\Theta_{\infty}(L_i+L_i^*)\big)Q^\rho_k(\infty)Q_v(\infty)\\
&=&Q^\rho_k(\infty)\tr(M_v(L_i\Theta_{\infty}+\Theta_{\infty} L_i^*)\end{eqnarray}
Furthermore
\begin{eqnarray}Q^\rho_v(\infty)\tr(M_k(L_i\Theta_{\infty}+\Theta_{\infty} L_i^*)&=&Q^\rho_v(\infty)Q^\rho_k(\infty)\tr(M_k(L_i\rho_k+\rho_k L_i^*)\\
&=&Q^\rho_v(\infty)Q^\rho_k(\infty)\tr(M_v(L_i\rho_v+\rho_v L_i^*)
\end{eqnarray}
In the same way we get for all $j=p+1,\ldots,n$
$$Q^\rho_v(\infty)Q^\rho_k(\infty)\tr(M_v(C_j\rho_vC_j^*)=Q_v(\infty)Q^\rho_k(\infty)\tr(M_k(C_j\rho_kC_j^*).$$
This implies that for all $v,k=1,\ldots,K$
$$Q^\rho_v(\infty)Q^\rho_k(\infty)=0.$$
At this stage since $\sum_{k=1}^KQ^\rho_k(\infty)=1$, there is only one index $u$ such that $Q^\rho_u(\infty)=1$ and all other are equal to zero. This determines the random variable $\Gamma$ which is valued in $\{1,\ldots,K\}$. 

 Now, in order to finish the proof, we just need to see that for all $k=1,\ldots,K$
$$\mathbb P^\rho(\Gamma=k)=\mathbb P^\rho(Q_k^\rho(\infty))=1=\mathbb E^\rho[Q_k^\rho(\infty)]=Q_k^\rho(0),$$
by the martingale property.
\end{proof}

\end{document}